\documentclass[12pt]{iopart}
\usepackage{iopams,mathptm,times,graphicx}
\makeatletter

\newcounter{theorem}
\@addtoreset{theorem}{section}
\renewcommand\thetheorem{\arabic{section}.\arabic{theorem}}

\newenvironment{lemma}{\par\medskip\noindent\begingroup{\bf Lemma
             \stepcounter{theorem}\thetheorem.}\ \itshape
             \def\@currentlabel{\thetheorem}}{\endgroup\par\medskip}
\newenvironment{theorem}{\par\medskip\noindent\begingroup{\bf Theorem
             \stepcounter{theorem}\thetheorem.}\ \itshape
             \def\@currentlabel{\thetheorem}}{\endgroup\par\medskip}

\newenvironment{remark}{\par\medskip\noindent\begingroup{\bf Remark
             \stepcounter{theorem}\thetheorem.}\
             \def\@currentlabel{\thetheorem}}{\endgroup\par\medskip}

\newenvironment{proof}{\par\noindent{\bf Proof.} }{\proofbox\par\medskip}

\def\proofbox{\hfill{\ensuremath\Box}}

\newdimen\LENB \newdimen\LENW \newdimen\THI
\newdimen\LENWH \newdimen\LENTOT \newcount\N
\def\vbrknlnele#1#2#3{
  \LENB=#1pt \LENW=#2pt \THI=#3pt
  \LENWH=\LENW \divide\LENWH by 2
  \LENTOT=\LENB \advance\LENTOT by \LENW
  \vbox to \LENTOT{
    \vbox to \LENWH{}
    \nointerlineskip
    \vbox to \LENB{\hbox to \THI{\vrule width \THI height \LENB}}
    \nointerlineskip
    \vbox to \LENWH{}
  }}

\def\vbrknln#1{
  \N=#1
  \vcenter{
    \vbox{
      \loop\ifnum\N>0
        \vbox to 4pt{\vbrknlnele{2}{2}{0.1}}
        \nointerlineskip
        \advance\N by -1
      \repeat
  }}}

\def\hbrknlnele#1#2#3{
  \LENB=#1pt \LENW=#2pt \THI=#3pt
  \LENTOT=\LENB \advance\LENTOT by \LENW
  \vcenter{
    \vbox to \THI{
      \hbox to \LENTOT{
        \hfil
        \vrule width \LENB height \THI
        \hfil}
  }}}

\makeatletter

\usepackage{ulem}

\def\journal#1&#2,{\begingroup \let\journal=\dummyjournal
               \it #1\unskip~\bf\ignorespaces #2\rm,\endgroup}
\def\dummyjournal{\errmessage{Reference foul up: nested \journal macros}}

\eqnobysec

\def\eqref#1{(\ref{#1})}

\begin{document}
\title[Integrable semi-discretizations of the reduced Ostrovsky equation]
  {Integrable semi-discretizations of the reduced Ostrovsky equation}
\author{Bao-Feng Feng$^1$, Ken-ichi Maruno$^2$ and
Yasuhiro Ohta$^{3}$
}
\address{$^1$~Department of Mathematics,
The University of Texas-Pan American,
Edinburg, TX 78541
}

\address{$^2$~Department of Applied Mathematics,
Waseda University, Tokyo 169-8050, Japan
}
\address{$^3$~Department of Mathematics,
Kobe University, Rokko, Kobe 657-8501, Japan
}
\eads{\mailto{feng@utpa.edu}, \mailto{kmaruno@waseda.jp} and
\mailto{ohta@math.kobe-u.ac.jp}}

\date{\today}
\def\submitto#1{\vspace{28pt plus 10pt minus 18pt}
     \noindent{\small\rm Accepted by : {\it #1}\par}}
\begin{abstract}
Based on our previous work to the reduced Ostrovsky equation (J. Phys. A {\bf 45} 355203), we
construct its integrable semi-discretizations. Since the reduced Ostrovsky equation admits two alternative representations, one is its original form, the other is the differentiation form, or the short wave limit of the  Degasperis-Procesi equation, two semi-discrete analogues of the reduced Ostrovsky equation are constructed possessing the same $N$-loop soliton solution. The relationship between these two versions of semi-discretizations is also clarified.
\par
\kern\bigskipamount\noindent
\today
\end{abstract}
\kern-\bigskipamount
\pacs{02.30.Ik, 05.45.Yv,42.65.Tg, 42.81.Dp}

\submitto{\JPA}

\section{Introduction}
In this paper, we consider integrable discretizations of the reduced Ostrovsky equation~
\begin{equation}
\partial_x
\left(\partial_t+u\partial_x\right)u
-3u=0\,,\label{vakhnenko}
\end{equation}
which is a special case ($\beta=0$) of the Ostrovsky equation
\begin{equation}
\partial_x
\left(\partial_t+u\partial_x+\beta \partial_x^3\right)u
-\gamma u=0\,.\label{ostrovsky}
\end{equation}
The Ostrovsky equation was originally derived as a model for weakly nonlinear surface and internal
waves in a rotating ocean~\cite{Ostrovsky,Stepanyants}.
Later on, the same equation was derived for different physical situations
by several authors~\cite{Hunter,Vakhnenko1}. Especially eq. (\ref{vakhnenko}) appears as a model for high-frequency
waves in a relaxing medium \cite{Vakhnenko1,Vakhnenko5}.
Note that the reduced Ostrovsky equation (\ref{vakhnenko}) is sometimes
called the Vakhnenko equation ~\cite{Vakhnenko2,Vakhnenko3,Vakhnenko4}, the Ostrovsky-Hunter equation
~\cite{Liu}, or the Ostrovsky-Vakhnenko equation~\cite{Brunelli,Shepelsky}. Travelling wave
solutions were investigated in \cite{Stepanyants,Boyd,Parkes}.
Vakhnenko et al. constructed the $N$ (loop) soliton
solution of the reduced
Ostrovsky equation by using a hodograph (reciprocal) transformation and
the Hirota bilinear method~\cite{Vakhnenko2,Vakhnenko3}.
The same problem was approached from the point of view of inverse
scattering method~\cite{Vakhnenko4}.

Differentiating the reduced Ostrovsky equation (\ref{vakhnenko}), we obtain
\begin{equation}
u_{txx}+3u_xu_{xx}+uu_{xxx}-3u_x=0\,,\label{short-DP-eq}
\end{equation}
which is known as the short wave limit of the Degasperis-Procesi (DP)
equation~\cite{Hone-Wang,MatsunoPLA}.
This equation is derived from the DP equation~\cite{DP}
\begin{equation}
U_T+3\kappa^3U_x- U_{TXX}+4UU_X=3U_XU_{XX}+UU_{XXX}\,,\label{DP-eq}
\end{equation}
by taking a short wave limit $\epsilon\to 0$ with
$U=\epsilon^2(u+\epsilon u_1+\cdots)$, $T=\epsilon t$,
$X=\epsilon^{-1}x$ and $\kappa=1$.
It is noted that the short wave limit of the DP equation can also be rewritten as
alternative form
\begin{equation}
(\partial_t+u\partial_x) m  = -3m u_x\,, \quad m=1-u_{xx}\,.
\end{equation}
Based on this connection, Matsuno \cite{MatsunoPLA} constructed $N$-soliton solution of the short wave model of
the DP equation from $N$-soliton solution of the DP equation~\cite{Matsuno-DP1,Matsuno-DP2}.
This $N$-soliton formula is equivalent to the one obtained by
Vakhnenko et al. ~\cite{Vakhnenko2,Vakhnenko3}.

As already mentioned previously, the reduced Ostrovsky equation (\ref{vakhnenko}), as well as its differentiation form (\ref{short-DP-eq}), has attracted much attention in the past. Hone and Wang constructed the Lax pairs for both of  equations \cite{Hone-Wang}. The bi-Hamiltonian structure for the reduced Ostrovsky equation (\ref{vakhnenko}) was found by Brunelli and Sakovich \cite{Brunelli},  its integrability and wave-breaking was studied in \cite{Grimshaw2012}. Interestingly, the short wave limit of the DP equation (\ref{short-DP-eq}) also serves as an asymptotic model for propagation of surface waves in deep water under the condition of small-aspect-ratio \cite{Manna2014}. Most recently, the inverse scattering transform (IST) problem for the short wave limit of the DP equation (\ref{short-DP-eq}) was solved by a Riemann-Hilbert approach \cite{Shepelsky}.

The reduced Ostrovsky equation (\ref{vakhnenko}) is known to be related to the Tzitzeica equation \cite{Tzitzeica1,Tzitzeica1,Willox,Nimmo09}, and also the so-called Dodd--Bullough--Mikhailov equation \cite{Dodd-Bullough,Mikhailov1,Mikhailov2}, by a reciprocal transformation. Based on this reciprocal link
 between the reduced Ostrovsky equation and 3-reduction of the B-type or C-type two-dimensional Toda lattice, i.e. the $A_2^{(2)}$ 2D-Toda lattice, multi-soliton solutions to both the reduced Ostrovsky equation (\ref{vakhnenko}) and its differentiation version were constructed by the authors \cite{FMO-VE}.

How to construct its integrable discrete analogue for a soliton equation has been an important topic since the discovery of soliton theory. Although several approaches have been developed starting from the mid-1970s, it remains a challenging and mysterious problem and has to be dealt with on a case by case base. Ablowitz and Ladik originated a method of integrable discretization based the Lax pair of a soliton equation \cite{AL1,AL2}. Almost at the same period, Hirota proposed an intriguing and universal approach based on the bilinear form of a soliton equation  \cite{Hirota-d1,Hirota-d2,Hirota-d3}. Another successful way to discretize soliton equations was proposed by Date, Jimbo and Miwa \cite{DJMdiscret1,DJMdiscret2,DJMdiscret3,DJMdiscret4,DJMdiscret5,Jimbo-Miwa} via the transformation group theory, which gives a large number of integrable disretizations. One of the most interesting example is the discrete KP equation, or the so-called Hirota-Miwa equation \cite{Hirota-d4,Jimbo-Miwa}, which can be viewed as the Master equation of discrete systems due to the reason that integrable discretization of many soliton equations such as discrete KdV equation, and discrete sine-Gordon equation can be obtained from Hirota-Miwa equation by reductions.
Suris also developed a general Hamiltonian approach for integrable discretizations of integrable systems, see Ref. \cite{Surisbook}.

The aim of this work is to construct integrable semi-discretizations of the reduced Ostrovsky equation
(\ref{vakhnenko}) and its differentiation form (\ref{short-DP-eq}) by virtue of Hirota's bilinear method.
The remainder of the present paper is organized as follows. In section 2, by constructing a semi-discrete analogue of a set of bilinear equations reduced from the period 3-reduction of the $B_{\infty}$ or
$C_{\infty}$ two-dimensional Toda system, we derive a semi-discrete reduced Ostrovsky equation based on Eq. (\ref{short-DP-eq}) and provide its $N$-loop soliton solution in terms of pfaffian. Then, an alternative semi-discrete reduced Ostrovsky equation is constructed based on Eq. (\ref{vakhnenko}) which shares the same $N$-loop soliton solution. It is interesting that a connection between two semi-discrete versions exists in analogue to a link between their continuous counterparts. We conclude our paper by some comments and further topics in section 4.
\section{Integrable semi-discretization of the short wave limit of the DP equation (\ref{short-DP-eq})}
It is shown in \cite{FMO-VE} that bilinear equations for the reduced Ostrovsky equation (\ref{short-DP-eq}) are
\begin{eqnarray}
&&-\left(\frac{1}{2}D_{y}D_{s}-1\right)f\cdot
 f=f g\,,\label{3CToda-bilinear1}\\
&&-\left(\frac{1}{2}D_{y}D_{s}-1\right)g \cdot
 g=f^2\,,
 \label{3CToda-bilinear2}
\end{eqnarray}
which originate from a period 3 reduction of BKP (CKP) hierarchy \cite{Nimmo09}.
Here $D_{y}D_{s}$ is the Hirota $D$-operator defined by
$$
D_y^m D_s^n f(y,s)\cdot g(y,s)=
\left(\frac{\partial}{\partial y} -\frac{\partial}{\partial {y'}} \right)^m
\left(\frac{\partial}{\partial s} -\frac{\partial}{\partial {s'}} \right)^n  f(y,s)g(y',s')|_{y=y',s=s'}\,.
$$
For the sake of convenience, we set $y=x_1$, $s=x_{-1}$. Under this reduction, one of the $tau$-functions $f$ turns out to be a square of a pfaffian
\cite{HirotaBook}
\begin{equation}\label{3CToda-bilinear3}
  \tau^2 = c f\,,
\end{equation}
where
$\tau={\rm Pf} (1,2, \cdots, 2N)$  is a pfaffian
whose elements are given by
\begin{equation}\label{pfaffian-VE}
  {\rm Pf} (i,j)=c_{i,j} + \frac{p_i-p_j}{p_i+p_j}e^{\xi_i+\xi_j}\,,
\end{equation}
with $$
 c_{i,j}=-c_{j,i}, \quad \xi_i=p_iy+ p_i^{-1} s + \xi_{i0}, \quad  c=\prod_{i=1}^{2N}2p_i\,.
$$

It was shown in \cite{FMO-VE} that bilinear equations (\ref{3CToda-bilinear1})--(\ref{3CToda-bilinear2}), together
with (\ref{3CToda-bilinear3}), yield the reduced Ostrovsky equation (\ref{short-DP-eq}) through a hodograph transformation
\begin{equation}\label{hodograph}
  x=y-2(\ln \tau)_{s}, \quad t=s,
\end{equation}
and a dependent transformation
\begin{equation}\label{u-transformation}
  u=-2(\ln \tau)_{ss}=-(\ln f)_{ss}\,.
\end{equation}
\begin{remark}
In accordance with the integrable discretizations which will be constructed hereafter, we choose an alternative hodograph transformation mentioned in Remark 2.15 of \cite{FMO-VE}.
\end{remark}
\subsection{Semi-discrete analogues of equations (\ref{3CToda-bilinear1})--(\ref{3CToda-bilinear3})}
Based on the results briefly mentioned above, we attempt to construct an integrable semi-discrete analogue of the
reduced Ostrovsky equation (\ref{short-DP-eq}). The key point is how to discretize the bilinear equations (\ref{3CToda-bilinear1})--(\ref{3CToda-bilinear3}). To this end, we start with Gram-type determinants
$$
g_{l}=\det_{1\le i,j\le 2N}\Big(m_{ij}(l)\Big), \qquad f_{l}=\det_{1\le i,j\le 2N}\Big(m_{ij}'(l)\Big)\,,
$$
where
$$
m_{ij}(l)=C_{i,j}+\frac{1}{p_i+p_j}\varphi_i^{(0)}(l)\varphi_j^{(0)}(l)\,,
$$
$$
m_{ij}'(l)=C_{i,j}+\frac{1}{p_i+p_j} \left( -\frac{p_j}{p_i} \right)\frac{1+bp_i}{1-bp_j}
\varphi_i^{(0)}(l)\varphi_j^{(0)}(l)\,,
$$
with
$$
C_{i,j}=C_{j,i}, \quad \varphi_i^{(n)}(l)
=p_i^n\left(\frac{1+bp_i}{1-bp_i}\right)^le^{\xi_i},
\quad
\xi_i=p_i^{-1} s+\xi_{i0}\,.
$$
Here $2b$ (not $b$) is the mesh size in $y$-direction. A relation between $f_l$ and $g_l$ is shown by the following lemma.
\begin{lemma}
\begin{equation}\label{Bilinear1}
 (D_s-2b)g_{l+1}\cdot g_{l}=-2bf^2_{l}\,,
\end{equation}
\end{lemma}
\begin{proof}
It can be easily verified that
\[
\partial_s m_{ij}(l)=\varphi_i^{(-1)}(l)\varphi_j^{(-1)}(l)\,,
\]
\[
m_{ij}(l+1)=m_{ij}(l)
+\frac{2b}{(1-bp_i)(1-bp_j)}\varphi_i^{(0)}(l)\varphi_j^{(0)}(l)\,,
\]
and
\[
m_{ij}'(l)=m_{ij}(l)
-\frac{1}{1-bp_j}\varphi_i^{(-1)}(l)\varphi_j^{(0)}(l)\,.
\]
Then by using the following formulas for a $N \times N$ determinant $M$ with $M_{ij}$ denoting the cofactor of the element $m_{ij}$
$$
\frac{\partial}{\partial s} |M| = \sum_{i,j=1}^N \frac{\partial m_{ij}}{\partial s} M_{ij},
\quad \left|\matrix{
m_{ij} &a_i \cr
b_j & d} \right| = d|M| - \sum_{i,j=1}^N a_ib_j M_{ij}  \,,
$$
we have
\begin{equation}
  \partial_s g_{l}=\left|\matrix{
m_{ij}(l) &\varphi_i^{(-1)}(l) \cr
-\varphi_j^{(-1)}(l) &0}\right|\,,
\label{rule1}
\end{equation}

\begin{equation}
 g_{l+1}=\left|\matrix{
m_{ij}(l) &\displaystyle\frac{2b}{1-bp_i}\varphi_i^{(0)}(l) \cr
\displaystyle -\frac{1}{1-bp_j}\varphi_j^{(0)}(l) &1}\right|\,,
\label{rule2}
\end{equation}

\begin{equation}
f_{l}=\left|\matrix{
m_{ij}(l) &\varphi_i^{(-1)}(l) \cr
\displaystyle\frac{1}{1-bp_j}\varphi_j^{(0)}(l) &1}\right|
=\left|\matrix{
m_{ij}(l) &\displaystyle\frac{1}{1-bp_i}\varphi_i^{(0)}(l) \cr
\varphi_j^{(-1)}(l) &1}\right|\,.
\label{rule3}
\end{equation}

Furthermore, we can show
\begin{eqnarray}
&&(\partial_s-2b)g_{l+1}=\left|\matrix{
m_{ij}(l) &\varphi_i^{(-1)}(l)
&\displaystyle\frac{2b}{1-bp_i}\varphi_i^{(0)}(l) \cr
-\varphi_j^{(-1)}(l) &0 &0 \cr
\displaystyle -\frac{1}{1-bp_j}\varphi_j^{(0)}(l) &0 &1}\right| \nonumber
\\&&\qquad
+\left|\matrix{
m_{ij}(l)
&\displaystyle \frac{2b(\partial_s-b)}{1-bp_i}\varphi_i^{(0)}(l) \cr
\displaystyle -\frac{1}{1-bp_j}\varphi_j^{(0)}(l) &-b}\right|
+\left|\matrix{
m_{ij}(l) &\displaystyle\frac{2b}{1-bp_i}\varphi_i^{(0)}(l) \cr
\displaystyle -\frac{(\partial_s-b)}{1-bp_j}\varphi_j^{(0)}(l)
&-b}\right| \nonumber
\\&&\qquad
=\left|\matrix{
m_{ij}(l) &\varphi_i^{(-1)}(l)
&\displaystyle\frac{2b}{1-bp_i}\varphi_i^{(0)}(l) \cr
-\varphi_j^{(-1)}(l) &0 &0 \cr
\displaystyle -\frac{1}{1-bp_j}\varphi_j^{(0)}(l) &0 &1}\right| \nonumber
\\&&\qquad
+\left|\matrix{
m_{ij}(l)
&\displaystyle 2b\varphi_i^{(-1)}(l) \cr
\displaystyle -\frac{1}{1-bp_j}\varphi_j^{(0)}(l) &-b}\right|
+\left|\matrix{
m_{ij}(l) &\displaystyle\frac{2b}{1-bp_i}\varphi_i^{(0)}(l) \cr
\displaystyle -\varphi_j^{(-1)}(l) &-b}\right| \nonumber
\\&&\qquad
=\left|\matrix{
m_{ij}(l) &\varphi_i^{(-1)}(l)
&\displaystyle\frac{2b}{1-bp_i}\varphi_i^{(0)}(l) \cr
-\varphi_j^{(-1)}(l) &0 &-2b \cr
\displaystyle -\frac{1}{1-bp_j}\varphi_j^{(0)}(l) &-1 &1}\right|\,.
\label{rule4}
\end{eqnarray}
By using the Jacobi's identity  of determinant and the relations (\ref{rule1})--(\ref{rule4}),
we obtain
\[
(\partial_s-2b)g_{l+1}\times g_{l}
=g_{l+1}\times\partial_sg_{l}-(-2bf_{l})\times(-f_{l})\,,
\]
which is nothing but Eq. (\ref{Bilinear1}).
\end{proof}
\begin{remark}
Eq. (\ref{Bilinear1}) is an integrable discretization of the bilinear equation (\ref{3CToda-bilinear3}) in $y$-direction. Note that $2b$ is the mesh size. In the limit of $b \to 0$, we have
$$
f_{l} \to f, \quad g_{l} \to g\,, \quad g_{l+1} \to g + 2b g_{y}\,,
$$
then it follows
$$
\frac{1}{2b} D_s g_{l+1}\cdot g_{l} \to \frac 12 D_sD_y g \cdot g\,.
$$
Therefore, equation (\ref{Bilinear1}) converges to equation (\ref{3CToda-bilinear3}) as $b \to 0$.
\end{remark}
Next, we perform reduction in order to mimic  the period 3-reduction of CKP/BKP hierarchy in the continuous case. To this end, we let $C_{i,j}$ take a special value as follows
\begin{equation}
C_{i,j}=\delta_{j,2N+1-i}c_i,
\quad
c_{2N+1-i}=c_i\,,
\end{equation}
and further assume
\begin{equation}
c_{i,j}=-C_{i,j} \frac{2p_i^2}{p_j}\frac{1-bp_j}{1+bp_i}\,.
\end{equation}
By imposing a reduction condition
\begin{equation}
p_i^3 (1-b^2p_{2N+1-i}^2) = - p_{2N+1-i}^3 (1-b^2p_i^2)\,,
\end{equation}
which can be written as
$$
\frac{p^2_i(1-bp_{2N+1-i})}{p_{2N+1-i}(1+bp_i)}
=-\frac{p_{2N+1-i}^2(1-bp_i)}{p_i(1+bp_{2N+1-i})}\,,
$$
it then follows
\begin{eqnarray*}
&&c_{i,j}=-\delta_{j,2N+1-i}c_i
\frac{2p_i^2}{p_{2N+1-i}} \frac{1-bp_{2N+1-i}}{1+bp_i}
\\&&\qquad
=\delta_{i,2N+1-j} c_{2N+1-i}
\frac{2p_{2N+1-i}^2}{p_i} \frac{1-bp_i}{1+bp_{2N+1-i}}
\\&&\qquad
=-c_{j,i}\,.
\end{eqnarray*}
Thus, we can define a pfaffian
$$
\tau_{l}= {\rm Pf} (1,2,\cdots,2N)_l\,
$$
whose elements are
$$
(i,j)_l=c_{i,j}
+\frac{p_i-p_j}{p_i+p_j}\varphi_i^{(0)}(l)\varphi_j^{(0)}(l)\,.
$$
The relations between the pfaffian $\tau_l$ and the Gram-type determinants $f_l$, $g_l$ are stated by the following lemma.
\begin{lemma}
\begin{equation}\label{Bilinear2}
 (D_s-b)\tau_{l+1}\cdot \tau_{l}=-bc' g_{l+1}\,,
\end{equation}
\begin{equation}\label{Bilinear3}
 \tau^2_{l}= c' f_l\,,
\end{equation}
where
\begin{equation*}
c'=\prod_{i=1}^{2N}2p_i\frac{1-bp_i}{1+bp_i}\,.
\end{equation*}
\end{lemma}
\begin{proof}
We firstly list two pfaffian identities which will be used in the process of proof
\begin{equation}\label{pfaffan-identity1}
\mathop {\rm Pf}_{1\le i< j \le 2N} \left(\delta \alpha_{ij} - a_i b_j + a_j b_i \right)
={\rm Pf}\pmatrix{
\matrix{\alpha_{ij}} &a_i & b_i \cr
&&\delta}\,,
\end{equation}
\begin{eqnarray}
 && \mathop {\rm det}_{1\le i, j \le 2N} \pmatrix{
\matrix{\alpha_{ij} } &a_i & b_i \cr
c_j & \alpha & \beta \cr
d_j & \gamma & \delta} = \mathop{\rm Pf} \left(\alpha \alpha_{ij} - a_i c_j + a_j c_i \right)
 \mathop{\rm Pf} \left(\delta \alpha_{ij} - b_i d_j + b_j d_i \right) \nonumber \\
   &&  \qquad  -\mathop{\rm Pf} \left(\gamma \alpha_{ij} - a_i d_j + a_j d_i \right)
   \mathop{\rm Pf} \left(\beta \alpha_{ij} - b_i c_j + b_j c_i \right)\,.
\label{pfaffan-identity2}
\end{eqnarray}
Since
$$
\partial_s (i,j)_{l}
=\varphi_i^{(0)}(l)\varphi_j^{(-1)}(l)
-\varphi_i^{(-1)}(l)\varphi_j^{(0)}(l)\,,
$$

$$
(i,j)_{l+1}=(i,j)_{l}
+\varphi_i^{(0)}(l+1)\varphi_j^{(0)}(l)
-\varphi_i^{(0)}(l)\varphi_j^{(0)}(l+1)\,,
$$

$$
(\partial_\tau-b) )(i,j)_{l+1}=-b (i,j)_{l}
+\varphi_i^{(0)}(l+1)\varphi_j^{(-1)}(l)
-\varphi_i^{(-1)}(l)\varphi_j^{(0)}(l+1)\,,
$$
we have
\begin{equation}
  \tau_{l}
={\rm Pf}\pmatrix{
\matrix{(i,j)_{l} \cr {}} &\varphi_i^{(0)}(l) &\varphi_i^{(0)}(l) \cr
&&1}\,,
\end{equation}

\begin{equation}
  \partial_s \tau_{l}
={\rm Pf}\pmatrix{
\matrix{(i,j)_{l} \cr {}} &\varphi_i^{(-1)}(l) &\varphi_i^{(0)}(l) \cr
&&0}\,,
\end{equation}

\begin{equation}
 \tau_{l+1}
={\rm Pf}\pmatrix{
\matrix{(i,j)_{l} \cr {}} &\varphi_i^{(0)}(l) &\varphi_i^{(0)}(l+1) \cr
&&1}\,,
\end{equation}

\begin{equation}
(\partial_s-b) \tau_{l+1}=
{\rm Pf}\pmatrix{
\matrix{(i,j)_{l} \cr {}} &\varphi_i^{(-1)}(l)
&\varphi_i^{(0)}(l+1) \cr
&&-b}\,,
\end{equation}
by referring to the identity (\ref{pfaffan-identity1}). Furthermore,  by using the pfaffian identify (\ref{pfaffan-identity2})
\begin{eqnarray}
\fl (D_s-b)\tau_{l+1}\cdot \tau_{l} && =\tau_l (\partial_s-b)\tau_{l+1}- \tau_{l+1} \partial_s \tau_{l} \nonumber \\ &&
={\rm Pf}\pmatrix{
\matrix{(i,j)_{l} \cr {}} &\varphi_i^{(-1)}(l)
&\varphi_i^{(0)}(l+1) \cr
&&-b} {\rm Pf}\pmatrix{
\matrix{(i,j)_{l} \cr {}} &\varphi_i^{(0)}(l) &\varphi_i^{(0)}(l+1) \cr
&&1} \nonumber
\\&& \qquad -{\rm Pf}\pmatrix{
\matrix{(i,j)_{l} \cr {}} &\varphi_i^{(0)}(l) &\varphi_i^{(0)}(l) \cr
&&1} {\rm Pf}\pmatrix{
\matrix{(i,j)_{l} \cr {}} &\varphi_i^{(-1)}(l) &\varphi_i^{(0)}(l) \cr
&&0} \nonumber
\\&&
=\left|\matrix{
m_{ij}(l) &\varphi_i^{(-1)}(l)
& \varphi_i^{(0)}(l) \cr
\varphi_j^{(0)}(l+1) & -b  &1 \cr
\varphi_j^{(0)}(l) &0 &1}\right| \nonumber
\\&&
=\left|\matrix{
m_{ij}(l)-\varphi_i^{(0)}(l)\varphi_j^{(0)}(l) &\varphi_i^{(-1)}(l)
& \varphi_i^{(0)}(l)\cr
\varphi_j^{(0)}(l+1)-\varphi_j^{(0)}(l) & -b  &1 \cr
0 &0 &1}\right| \nonumber
\\&&  =\left|\matrix{
m_{ij}(l)-\varphi_i^{(1)}(l)\varphi_j^{(1)}(l)   &\varphi_i^{(-1)}(l)
 \cr
b( \varphi_j^{(0)}(l+1)+\varphi_j^{(0)}(l))   & -b}\right| \nonumber
\\&&
= -b \det\left(
m_{ij}(l)-\varphi_i^{(0)}(l)\varphi_j^{(0)}(l) + \varphi_i^{(-1)}(l) (\varphi_j^{(1)}(l+1)+\varphi_j^{(1)}(l))\right)\nonumber
\\&&
= -b \det\left(
c_{i,j}+ \frac{2p^2_j (1+bp_i)}{p_i (p_i+p_j) (1-bp_j) }\varphi_i^{(0)}(l)\varphi_j^{(0)}(l)\right) \nonumber
\\&&
= -bc' \det\left(
\frac{p_i (1-bp_j) }{2p^2_j (1+bp_i) } c_{i,j}+ \frac{1}{p_i+p_j}\varphi_i^{(0)}(l)\varphi_j^{(0)}(l)\right) \nonumber
\\&&
= -bc' \det(
 C_{i,j}+ \frac{1}{p_i+p_j}\varphi_i^{(0)}(l)\varphi_j^{(0)}(l))  \nonumber
\\&&
= -bc' F_{l+1}\,.
\end{eqnarray}
Thus Eq. (\ref{Bilinear2}) is proved. Next, we prove the relation (\ref{Bilinear3}).
\begin{eqnarray}
&&f_l=\det\left(\delta_{j,2N+1-i} c_i-\frac{p_j}{p_i}\frac{1}{p_i+p_j}\frac{1+bp_i}{1-bp_j}
\varphi_i^{(0)}(l)\varphi_j^{(0)}(l)\right) \nonumber
\\
&&=\prod_{i=1}^{2N}\frac{1+bp_i}{2p_i(1-bp_i)} \det\left(-\delta_{j,2N+1-i}c_i\frac{2p_i^2}{p_j}\frac{1-bp_j}{1+bp_i}+\frac{2p_i}{p_i+p_j}
\varphi_i^{(0)}(l)\varphi_j^{(0)}(l)\right) \nonumber \\
&&=\frac{1}{c'} \det\left(c_{i,j}+\left( \frac{p_i-p_j}{p_i+p_j}-1\right)
\varphi_i^{(0)}(l)\varphi_j^{(0)}(l)\right) \nonumber \\
&& =  \frac{1}{c'} \left[ {\rm Pf} \Big(c_{i,j}+
\frac{p_i-p_j}{p_i+p_j} \varphi_i^{(0)}(l)\varphi_j^{(0)}(l)
\Big)\right]^2\,.
\end{eqnarray}
Therefore Eq. (\ref{Bilinear3}) holds.
\end{proof}
\begin{remark}
Obviously, Eq. (\ref{Bilinear3}) converges to (\ref{3CToda-bilinear3}) as $b \to 0$ since $\tau_l \to \tau$, $f_l \to f$ and $c' \to c$ under this limit.
\end{remark}
\begin{remark}
Multiplying both sides of Eq. (\ref{Bilinear2}) by $2\tau_l \tau_{l+1}$, we have
\begin{equation*}
 (D_s-2b)\tau^2_{l+1}\cdot \tau^2_{l}=-2bc' g_{l+1}\tau_{l+1} \tau_l \,
\end{equation*}
by using a bilinear identify $D_s f^2 \cdot g^2=2fg D_s f \cdot g$. Furthermore, by referring to the relation (\ref{3CToda-bilinear3}), we have
\begin{equation*}
 \left(\frac{1}{2b} D_s-1\right)f^2_{l+1}\cdot f^2_{l}=-\frac{1}{c'} g_{l+1}\tau_{l+1} \tau_l \,
\end{equation*}
which converges to (\ref{3CToda-bilinear1}) as $b \to 0$ since $g_{l+1} \to g$ and $\tau_{l+1} \tau_l/c' \to \tau^2/c =f$ under this limit.
\end{remark}
\subsection{Integrable semi-discretization of the short wave limit of the DP equation (\ref{short-DP-eq})}
Summarizing what we have discussed in the previous subsection, the following three relations
\begin{eqnarray}
&& (D_s-2b)g_{l+1}\cdot g_{l} =  -2b f^2_{l}\,, \label{VE-sdbilinear1} \\
   && (D_s-b)\tau_{l+1}\cdot \tau_{l}=-bc' g_{l+1}\,, \label{VE-sdbilinear2}\\
   && \tau^2_l = c' f_l \label{VE-sdbilinear3} \,.
\end{eqnarray}
constitute the semi-discrete analogue of bilinear equations  (\ref{3CToda-bilinear1})--(\ref{3CToda-bilinear3}).
Let us construct integrable semi-discretization of the reduced Ostrovsky equation based on bilinear equations (\ref{VE-sdbilinear1})--(\ref{VE-sdbilinear3}).
First, we rewrite Eqs. (\ref{VE-sdbilinear1}) and (\ref{VE-sdbilinear2}) into
\begin{equation}
\left(\ln\frac{g_{l+1}}{g_l}\right)_s-2b=-2b\frac{f^2_l}{g_{l+1}g_l}\,,
\label{BL4}
\end{equation}
and \begin{equation}
\left(\ln\frac{\tau_{l+1}}{\tau_l}\right)_s-b=-bc' \frac{g_{l+1}}{\tau_{l+1}\tau_l}\,,
\label{BL5}
\end{equation}
respectively.
Introducing a discrete  hodograph transformation
\begin{equation}
\label{sd_hodograph_trf}
x_l=2lb-2(\ln \tau_l)_s, \quad t=s\,,
\end{equation}
and a dependent variable transformation
\begin{equation}
\label{sd_u_trf}
u_l= -2 (\ln \tau_l)_{ss}=- (\ln f_l)_{ss}\,,
\end{equation}
it then follows that the nonuniform mesh, which is defined by $\delta_l= x_{l+1}-x_l$, can be expressed as
\begin{equation}\label{mesh}
\delta_l= 2b -2 \left(\ln \frac{\tau_{l+1}}{\tau_l} \right)_s =2bc' \frac{g_{l+1}}{\tau_{l+1}\tau_l}\,,
\end{equation}
with the use of (\ref{BL5}). Differentiating  Eq. (\ref{mesh}) with respect to $s$, one obtains
\begin{equation}\label{sd1_VE2}
    \frac{d \delta_l}{d s} = -2\left(\ln \frac{\tau_{l+1}}{\tau_l} \right)_{ss} = u_{l+1}-u_{l}\,.
\end{equation}
Introducing an auxiliary variable $r_l = f_l/g_l$, we then have
\begin{equation}\label{sd_VEtt}
 \frac{4}{\delta^2_l} =  \frac{1}{b^2} \frac{g_{l}}{g_{l+1}} r_lr_{l+1}\,,
\end{equation}
where (\ref{VE-sdbilinear3}) is used. Further, one obtains
\begin{equation}
  \left(\ln \frac{r_{l+1}}{r_l} \right)_{s} +\delta_{l}= 2b \frac{f^2_l}{g_{l+1}g_l}\,
  \label{sd-VEt3}
\end{equation}
 by referring to (\ref{BL5}) and (\ref{BL4}).
Taking the logarithmic derivative of (\ref{sd_VEtt}) with respect to $s$ leads to
\begin{equation} \label{sd_VEttt}
 \left(\ln {r_{l+1}}{r_l} \right)_{s}- \left(\ln \frac{g_{l+1}}{g_l} \right)_{s} = - \frac{2}{\delta_l} \frac{d \delta_l}{d s}\,.
\end{equation}
Substituting  Eq.(\ref{sd1_VE2}) into Eq. (\ref{sd_VEttt}) and referring  to Eqs. and (\ref{BL4}) and (\ref{sd-VEt3}), one obtains
\begin{eqnarray}
   \frac{u_{l+1}-u_{l}}{\delta_l}&=& -\frac 12 \left(\ln {r_{l+1}}{r_l} \right)_{s} + \frac 12 \left(\ln \frac{g_{l+1}}{g_l} \right)_{s} \nonumber \\
   &=& -\frac 12 \left(\ln r_{l+1}{r_l} \right)_{s} + b - b \frac {f^2_{l}}{g_l g_{l+1}} \nonumber \\
   &=&  -\frac 12 \left(\ln {r_{l+1}}{r_l} \right)_{s} + b - \frac 12 \left(\ln \frac {r_{l+1}}{r_l} \right)_{s} - \frac 12 \delta_l  \nonumber\\
   &=& - \left(\ln r_{l+1} \right)_{s} + b - \frac 12 \delta_l\,,
\end{eqnarray}
which can be recast into
\begin{equation}\label{sd-VEt2}
    \left(\ln r_{l+1} \right)_{s}= -\frac{u_{l+1}-u_{l}}{\delta_l} + b - \frac 12 \delta_l\,.
\end{equation}
A substitution of (\ref{sd-VEt2}) back into (\ref{sd-VEt3}) leads to
\begin{equation}\label{sd-VEt4}
   -\frac{u_{l+1}-u_{l}}{\delta_l} + \frac{u_{l}-u_{l-1}}{\delta_{l-1}} +  \frac 12 \delta_l + \frac 12 \delta_{l-1} = \frac{8b^3} {\delta^2_{l}} \frac{r_{l}}{r_{l+1}}\,.
\end{equation}
Defining
$$
m_l = \frac{2}{\delta_{l}+\delta_{l-1}} \left( -\frac{u_{l+1}-u_{l}}{\delta_l} + \frac{u_{l}-u_{l-1}}{\delta_{l-1}} \right) + 1\,,
$$
and taking the logarithmic derivative on both sides of (\ref{sd-VEt4}), we have
\begin{eqnarray}
  && \frac{d\,\ln m_l}{d\,s}  = (\ln r_l)_s - (\ln r_{l+1})_s -\frac {2}{\delta_l} \frac{d \delta_l}{d s}
  - \frac{d}{d s} (\delta_l+\delta_{l-1}) \nonumber \\
   &=&   -\frac{u_{l}-u_{l-1}}{\delta_{l-1}}  - \frac 12 \delta_{l-1}
    + \frac{u_{l+1}-u_{l}}{\delta_{l}}  + \frac 12 \delta_{l}  -\frac {2(u_{l+1}-u_{l})}{\delta_l} -\frac{u_{l+1}-u_{l-1}}{\delta_l+\delta_{l-1}}\nonumber \\
   &=&  -\frac{u_{l}-u_{l-1}}{\delta_{l-1}}
    - \frac{u_{l+1}-u_{l}}{\delta_{l}} -\frac{u_{l+1}-u_{l-1}}{\delta_l+\delta_{l-1}}- \frac 12 (\delta_{l}-\delta_{l-1}) \,.
\end{eqnarray}

As a result, by defining forward difference and average operators
$$
\Delta u_l = \frac{u_{l+1}-u_l}{\delta_l}, \quad M u_l =\frac{u_{l}+u_{l-1}}{2}\,,
$$
we can summarize what we have deduced into the following theorem.
\begin{theorem}
The semi-discrete analogue  of the short wave limit of the DP equation
\begin{equation}
\label{sd1-VE}
\left\{
\begin{array}{l} \displaystyle
 \frac{d\,m_l}{d \,s}
  = m_l \left( -2M \Delta u_l -\frac{M (\delta_l \Delta u_l)}{M \delta_l} -  \frac 12 (\delta_{l}  -\delta_{l-1})\right)\,, \\
\displaystyle \frac{d \, \delta_l}{d\, s} = u_{l+1}-u_{l}\, \\
\displaystyle m_l= -2\frac{M \Delta u_l}{M \delta_l}+1\,,
\end{array}\right.
\end{equation}
is determined from the following equations
\begin{equation*}
\left\{
\begin{array}{l} \displaystyle
 (D_s-2b)g_{l+1}\cdot g_{l} =  -2b f^2_{l}\,, \\
\displaystyle (D_s-b)f_{l+1}\cdot f_{l}=-bc' g_{l+1}\,,\\
\displaystyle \tau^2_l = c' f_l  \,,
\end{array}\right.
\end{equation*}
through discrete hodograph transformation $x_l=2lb -2 \left(\ln {\tau_l} \right)_{s}$, $\delta_l=x_{l+1}-x_l$, $t=s$ and dependent variable transformation $u_l= -2 (\ln \tau_l)_{ss}=- (\ln f_l)_{ss}$.
\end{theorem}

Let us consider the continuous limit when $b \to 0$.
The dependent variable $u$ is a function of $l$ and $s$. Meanwhile, we regard
it as a function of $x$ and $t$, where $x$ is the space coordinate
at $l$-th lattice point and $t$ is the time, defined by
$$
x=x_0+\sum_{j=0}^{l-1}\delta_j\,,\qquad t=s
$$
Then in the continuous limit, $b \to 0$ ($\delta_l \to 0$), we have
$$
2M \Delta u_l= \frac{u_{l+1}-u_l}{\delta_l}+ \frac{u_{l}-u_{l-1}}{\delta_{l-1}} \to 2u_x \,,
\quad \frac{M (\delta_l \Delta u_l)}{M \delta_l}= \frac{u_{l+1}-u_{l-1}}{\delta_l+\delta_{l-1}} \to u_x\,,
$$
$$
m_l=\frac{2}{\delta_{l}+\delta_{l-1}} \left( -\frac{u_{l+1}-u_{l}}{\delta_l} + \frac{u_{l}-u_{l-1}}{\delta_{l-1}} \right) + 1  \to  m=- u_{xx}+1\,.
$$
Moreover, since
$$
\frac{\partial x}{\partial s}
=\frac{\partial x_0}{\partial s}
+\sum_{j=0}^{l-1}\frac{\partial\delta_j}{\partial s}
=\frac{\partial x_0}{\partial s}
+\sum_{j=0}^{l-1}(u_{j+1}-u_j)
\to u \,,
$$
we then have
$$
\partial_s=\partial_t+\frac{\partial x}{\partial s} \partial_x
\to \partial_t+u\partial_x \,.
$$
Consequently, the third equation in (\ref{sd1-VE}) converges to $m=1-u_{xx}$. Whereas the first equation in (\ref{sd1-VE}) converges to
\begin{equation}
(\partial_t+u\partial_x) m  = -3m u_x\,,
\end{equation}
which is exactly the short wave limit of the DP equation (\ref{short-DP-eq}).

Based on the results in previous section, we can provide $N$-soliton solution to the semi-discrete reduced Ostrovsky equation
\begin{theorem}
The $N$-soliton solution to the semi-discrete analogue for the short wave limit of the DP equation (\ref{sd1-VE}) takes the following parametric form
$$u_l= -2 (\ln \tau_l)_{ss}, \quad  x_l=2lb -2 \left(\ln {\tau_l} \right)_{s}$$
where $\tau_l$ is a pfaffian
\begin{equation}
\label{Nsoliton-VE1}
\tau_{l}= {\rm Pf} (1,2,\cdots,2N)_l\,,
\end{equation}
whose elements are
\begin{equation}
\label{Nsoliton-VE2}
\fl (i,j)_l=c_{i,j}
+\frac{p_i-p_j}{p_i+p_j}\varphi_i^{(0)}(l)\varphi_j^{(0)}(l)\,, \quad  \varphi_i^{(n)}(l)
=p_i^n\left(\frac{1+bp_i}{1-bp_i}\right)^le^{p_i^{-1} s+\xi_{i0}}
\end{equation}
under the reduction condition
\begin{equation}
\label{Nsoliton-VE3}
p_i^3 (1-b^2p_{2N+1-i}^2) = - p_{2N+1-i}^3 (1-b^2p_i^2)\,, \quad i=1, 2, \cdots, N\,.
\end{equation}
\end{theorem}

\section{Integrable semi-discretization of the reduced Ostrovsky equation (\ref{vakhnenko})}
\subsection{Bilinear equation for the reduced Ostrovsky equation (\ref{vakhnenko})}
In this section, we will deduce an integrable semi-discrete analogue to the reduced Ostrovsky equation (\ref{vakhnenko}). It was pointed out by the authors \cite{FMO-VE} that a single bilinear equation (2.53) yields the reduced Ostrovsky equation (\ref{vakhnenko}). In order to be consistent with the $N$-soliton solution given in the previous section, we start with
\begin{equation}
[(D_{x_{-3}}-D_{x_{-1}}^3)D_{x_{1}}+3D_{x_{-1}}^2]\tau\cdot \tau=0\,,\label{BKP-bilinear}
\end{equation}
which is a dual bilinear equation of (2.47) in \cite{FMO-VE} for the extended BKP hierarchy.
Imposing the same period 3-reduction by requesting $D_{x_3}=D_{x_{-3}}=0$ and assuming $y=x_1$, $s=x_{-1}$, Eq. (\ref{BKP-bilinear}) is reduced to
\begin{equation} \label{vakhnenko-bilinear}
(D_{y}D_{s}^3-3D_{s}^2)\tau \cdot \tau=0\,.
\end{equation}

Prior to proceeding to the semi-discretization of Eq. (\ref{vakhnenko}), let us briefly show how the reduced Ostrovsky equation (\ref{vakhnenko}) is derived from Eq. (\ref{vakhnenko-bilinear}) through the same  hodograph transformation (\ref{hodograph}) and dependent variable transformation (\ref{u-transformation}) defined in the previous section. By defining $\rho^{-1}= 1-2(\ln \tau)_{ys}$, a conversion formula
\begin{equation}
\left\{
\begin{array}{l} \displaystyle
\frac{\partial}{\partial y} =  \frac{1}{\rho} \frac{\partial}{\partial x}\,,\\
\displaystyle
\frac{\partial}{\partial s}=\frac{\partial}{\partial t}+u\frac{\partial}{\partial x} \,,\qquad \qquad \quad
\end{array}
\right.\label{hodograph2}
\end{equation}
can be easily obtained from the hodograph transformation (\ref{hodograph}).
By using the relations
\begin{eqnarray*}
  \frac{D_{y}D_{s}^3 \tau \cdot \tau}{\tau^2} &=& 2 (\ln \tau)_{ysss} + 12 (\ln \tau)_{ss} (\ln \tau)_{ys}\,, \\
  \frac{D_{s}^2\tau \cdot \tau}{\tau^2} &=& 2(\ln \tau)_{ss}\,,
\end{eqnarray*}
Eq. (\ref{vakhnenko-bilinear}) is converted to
\begin{equation}
2 (\ln \tau)_{ysss} = 6 (\ln \tau)_{ss} (1-2(\ln \tau)_{ys})\,,
\end{equation}
 and is further reduced to
\begin{equation}
\rho u_{x_1x_{-1}} = 3u\,.
\end{equation}
With the use of the conversion formulas (\ref{hodograph2}), we finally arrive at
\begin{equation}
\partial_x (\partial_t+u\partial_x)u = 3u\,,
\end{equation}
which is exactly the reduced Ostrovsky equation (\ref{vakhnenko}).
\subsection{Semi-discrete analogue of the reduced Ostrovsky equation (\ref{vakhnenko})}
In order to obtain a discrete analogue for the bilinear equation (\ref{vakhnenko-bilinear}), we first prove a bilinear equation associated with the modified BKP.
 \begin{lemma}
 Assume a pfaffian $\tau_{l}= {\rm Pf} (1,2,\cdots,2N)_l$ with element determined by
 \begin{equation}
\label{bilinear_BKP1}
(i,j)_l=c_{i,j}
+\frac{p_i-p_j}{p_i+p_j}\varphi_i^{(0)}(l)\varphi_j^{(0)}(l)\,,
\end{equation}
where
$$
\varphi_i^{(n)}(l)
=p_i^n\left(\frac{1+bp_i}{1-bp_i}\right)^le^{\xi_i},
\quad
\xi_i=p_i^{-1} s+p_i^{-3} r+\xi_{i0}\,.
$$
Then the pfaffian $\tau_{l}$ satisfies the following bilinear equation
 \begin{equation}
 \left((D_s-b)^3-(D_r-b^3)\right)\tau_{l+1}\cdot\tau_l=0\,.
\label{bilinear_BKP2}
\end{equation}
 \end{lemma}
 \begin{proof}
First, we define the pfaffian elements in addition to (\ref{bilinear_BKP1}):
$$
{\rm Pf}(i,d_n)_{l}=\varphi_i^{(n)}(l)\,,
\qquad
{\rm Pf}(d_m,d_n)_{l}=0\,,
$$
$$
{\rm Pf}(i,d^l)_{l}=\varphi_i^{(0)}(l+1)\,,
\qquad
(d_m,d^k)_{kl}=(-b)^{-m}\,.
$$
Then the following differential and difference formulas are obtained previously or can be easily verified
$$
  \partial_s \tau_{l}
={\rm Pf}\pmatrix{
\matrix{(i,j)_{l} \cr {}} &\varphi_i^{(-1)}(l) &\varphi_i^{(0)}(l) \cr
&&0}={\rm Pf}(1,2,\cdots,2N,d_{-1},d_0)_{l}\,,
$$
$$
  \partial_s^2 \tau_{l}
={\rm Pf}\pmatrix{
\matrix{(i,j)_{l} \cr {}} &\varphi_i^{(-2)}(l) &\varphi_i^{(0)}(l) \cr
&&0}={\rm Pf}(1,2,\cdots,2N,d_{-2},d_0)_{l}\,,
$$

$$
   \tau_{l+1}
={\rm Pf}\pmatrix{
\matrix{(i,j)_{l} \cr {}} &\varphi_i^{(0)}(l) &\varphi_i^{(0)}(l+1) \cr
&&1}={\rm Pf}(1,2,\cdots,2N,d_{0},d^l)_{l}\,,
$$

$$
(\partial_s-b) \tau_{l+1}=
{\rm Pf}\pmatrix{
\matrix{ (i,j)_{l} \cr {}} &\varphi_i^{(-1)}(l)
&\varphi_i^{(0)}(l+1) \cr
&&-b}={\rm Pf}(1,2,\cdots,2N,d_{-1},d^l)_{l}\,,
$$

$$
(\partial_s-b)^2 \tau_{l+1}=
{\rm Pf}\pmatrix{
\matrix{ (i,j)_{l} \cr {}} &\varphi_i^{(-2)}(l)
&\varphi_i^{(0)}(l+1) \cr
&&b^2}={\rm Pf}(1,2,\cdots,2N,d_{-2},d^l)_{l}\,,
$$

$$
  \frac{1}{3}(\partial_s^3-\partial_r) \tau_{l}
={\rm Pf}\pmatrix{
\matrix{(i,j)_{l} \cr {}} &\varphi_i^{(-2)}(l) &\varphi_i^{(-1)}(l) \cr
&&0}={\rm Pf}(1,2,\cdots,2N,d_{-2},d_{-1})_{l}\,.
$$
Moreover, the following relations can be further verified
\begin{eqnarray*}
&&  (\partial_r-b^3) \tau_{l+1} = {\rm Pf}\pmatrix{
\matrix{ (i,j)_{l} \cr {}} &\varphi_i^{(-3)}(l)
&\varphi_i^{(0)}(l+1) \cr
&&-b^3} \\
   &&  \qquad -2{\rm Pf}\pmatrix{
\matrix{ (i,j)_{l} \cr {}} &\varphi_i^{(-2)}(l) &\varphi_i^{(-1)}(l)
&\varphi_i^{(0)}(l) &\varphi_i^{(0)}(l+1) \cr
&&0&0&b^2 \cr
&&&0&-b \cr
&&&&1}\\
&& \qquad ={\rm Pf}(1,2,\cdots,2N,d_{-1},d^l)_{l}-2{\rm Pf}(1,2,\cdots,2N,d_{-2},d_{-1},d_{0},d^l)_{l}\,,
\end{eqnarray*}
\begin{eqnarray*}
   && (\partial_s-b)^3 \tau_{l+1}=
{\rm Pf}\pmatrix{
\matrix{ (i,j)_{l} \cr {}} &\varphi_i^{(-3)}(l)
&\varphi_i^{(0)}(l+1) \cr
&&-b^3} \\
   &&  \qquad +{\rm Pf}\pmatrix{
\matrix{(i,j)_{l} \cr {}} &\varphi_i^{(-2)}(l) &\varphi_i^{(-1)}(l)
&\varphi_i^{(0)}(l) &\varphi_i^{(0)}(l+1) \cr
&&0&0&b^2 \cr
&&&0&-b \cr
&&&&1}\\
   &&  \qquad ={\rm Pf}(1,2,\cdots,2N,d_{-1},d^l)_{l}+{\rm Pf}(1,2,\cdots,2N,d_{-2},d_{-1},d_{0},d^l)_{l}\,,
\end{eqnarray*}
thus, we get
\begin{equation*}
\frac{1}{3}\left((\partial_s-b)^3-(\partial_r-b^3)\right) \tau_{l+1}=
{\rm Pf}(1,2,\cdots,2N,d_{-2},d_{-1},d_{0},d^l)_{l}\,.
\end{equation*}
Then an algebraic identity of pfaffian \cite{HirotaBook}
\begin{eqnarray*}
  && {\rm Pf} (\cdots, d_{-2}, d_{-1}, d_0, d^l) {\rm Pf} (\cdots)= {\rm Pf} (\cdots, d_{-2}, d_{-1}) {\rm Pf} (\cdots, d_0, d_1) \\
   && \quad - {\rm Pf} (\cdots, d_{-2},  d_0) {\rm Pf} (\cdots, d_{-1}, d^l) +
   {\rm Pf} (\cdots, d_{-2}, d^l) {\rm Pf} (\cdots, d_0, d_1)\,,
\end{eqnarray*}
derives
$$
\frac{1}{3}\left((\partial_s-b)^3-(\partial_r-b^3)\right) \tau_{l+1} \times \tau_{l}=
  \frac{1}{3}(\partial_s^3-\partial_r) \tau_{l} \times \tau_{l+1}-   \partial_s^2 \tau_{l} \times
(\partial_s-b) \tau_{l+1}+ (\partial_s-b)^2 \tau_{l+1} \times \partial_s \tau_{l}\,,
$$
which is equivalent to
$$
\left((D_s-b)^3-(D_r-b^3)\right)\tau_{l+1}\cdot\tau_l=0\,.
$$
 \end{proof}
Next, we preform a reduction in parallel to period 3 reduction for the continuous case.
Imposing the same reduction condition (\ref{Nsoliton-VE3}), which is also of the form
$$
\frac{1}{p_i^3}+\frac{1}{p_{2N+1-i}^3}
=b^2\left(\frac{1}{p_i}+\frac{1}{p_{2N+1-i}}\right)\,,
$$
and note that $\tau_l$ is rewritten as
$$
\tau_{l}= \left(\prod_{i=1}^{2N}\varphi_i^{(0)}(l)\right)
{\rm Pf}\pmatrix{\displaystyle
\frac{\delta_{j,2N+1-i}c_{i,j}}{\varphi_i^{(0)}(l)\varphi_{2N+1-i}^{(0)}(l)}
+\frac{p_i-p_j}{p_i+p_j}}\,,
$$
it can easily shown that the pfaffian $\tau_l$ satisfies
\begin{equation}
 \partial_r\tau_l=b^2\partial_s\tau_l\,,
 \end{equation}
therefore we have
 \begin{equation}
   (D_{s}^3-3bD_{s}^2+2b ^2D_{s})\tau_{l+1} \cdot \tau_l=0\,.
   \label{Bilinear_semiVE}
 \end{equation}
 In what follows, we construct a semi-discrete reduced Ostrovsky equation based on Eq. (\ref{Bilinear_semiVE}). First, by using the following relations
 \begin{equation*}
  \frac{D_{s} \tau_{l+1} \cdot \tau_{l}}{\tau_{l+1}\tau_{l}}= \left(\ln  \frac{\tau_{l+1}}{\tau_{l}}\right)_s\,,
 \end{equation*}

  \begin{equation*}
  \frac{D^2_{s} \tau_{l+1} \cdot \tau_{l}}{\tau_{l+1}\tau_{l}} = \left(\ln (\tau_{l+1}\tau_{l})\right)_{ss} +\left(\left(\ln  \frac{\tau_{l+1}}{\tau_{l}}\right)_s\right)^2  \,,
 \end{equation*}

  \begin{equation*}
 \frac{D^3_{s} \tau_{l+1} \cdot \tau_{l}}{\tau_{l+1}\tau_{l}} = \left(\ln  \frac{\tau_{l+1}}{\tau_{l}}\right)_{sss}+ 3
\left(\ln  \frac{\tau_{l+1}}{\tau_{l}}\right)_s \left(\ln (\tau_{l+1}\tau_{l})\right)_{ss} + \left(\left(\ln  \frac{\tau_{l+1}}{\tau_{l}}\right)_s\right)^3  \,,
 \end{equation*}
 one obtains
 \begin{eqnarray}
   \left(\ln  \frac{\tau_{l+1}}{\tau_{l}}\right)_{sss} &=&    \left( b-\left(\ln  \frac{\tau_{l+1}}{\tau_{l}}\right)_s\right) \nonumber \\
    &&  \left[3\left(\ln (\tau_{l+1}\tau_{l})\right)_{ss} - \left(\ln  \frac{\tau_{l+1}}{\tau_{l}}\right)_s \left( 2b-\left(\ln  \frac{\tau_{l+1}}{\tau_{l}}\right)_s\right)  \right]\,, \label{BL_semiVE2}
 \end{eqnarray}
 from  Eq. (\ref{Bilinear_semiVE}). Next, by using the discrete hodograph transformation (\ref{sd_hodograph_trf}) and dependent variable transformation (\ref{sd_u_trf}), Eq. (\ref{BL_semiVE2}) reads
\begin{equation}\label{sd2-VE1}
  \frac{d} {ds} (u_{l+1}-u_l) = \frac 32 \delta_l (u_{l}+u_{l+1}) -\frac 14 \delta_l(\delta^2_l-4b^2) \,.
\end{equation}
Obviously, the evolution equation for nonuniform mesh $\delta_l$ remains the same as Eq. (\ref{sd1_VE2}). In summary
 \begin{theorem}
 The bilinear equation
 \begin{equation*}
(\frac {1}{b} D_{s}^3-3D_{s}^2+2b D_{s})\tau_{l+1} \cdot \tau_l=0\,
\end{equation*}
determines a semi-discrete analogue of the reduced Ostrovsky equation (\ref{vakhnenko})
\begin{equation}
\label{sd2-VE}
\left\{
\begin{array}{l} \displaystyle
  \frac{d} {ds} (u_{l+1}-u_l) = \frac 32 \delta_l (u_{l}+u_{l+1}) -\frac 14 \delta_l(\delta^2_l-4b^2) \,, \\
\displaystyle \frac{d \, \delta_l}{d\, s} = u_{l+1}-u_{l}\,. \\
\end{array}\right.
\end{equation}
through the dependent variable transformation
$u_l= -2 (\ln \tau_l)_{ss}$ and the discrete hodograph transformation $x_l=2lb -2 \left(\ln {\tau_l} \right)_{s}$, $\delta_l=x_{l+1}-x_l$.
  \end{theorem}
Now we turn to check if Eq. (\ref{sd2-VE}) converges to Eq. (\ref{vakhnenko}) in the continuous limit.
By dividing $\delta_l$ on both sides of Eq. (\ref{sd2-VE1}), we have
\begin{equation}
\label{sd2-VE3}
  \frac{1}{\delta_l} \frac{d} {ds} (u_{l+1}-u_l) = \frac 32  (u_{l}+u_{l+1}) -\frac 14 \delta^2_l +b^2\,,
\end{equation}
which converges to exactly the reduced Ostrovsky equation (\ref{vakhnenko})
  \begin{equation*}
  \partial_x (\partial_t+u\partial_x ) u = 3u \,,
\end{equation*}
as $b \to 0$ ($\delta_l \to 0$).

Regarding the $N$-soliton solution, it is obvious that Eq. (\ref{sd2-VE}) admits the same solution as the semi-discrete reduced Ostrovsky equation (\ref{sd1_VE2}) proposed previously.

So far, we have constructed semi-discrete analogues of the reduced Ostrovsky equation (\ref{vakhnenko}) and its differentiation form (\ref{short-DP-eq}). In light of the link between (\ref{short-DP-eq}) and (\ref{vakhnenko}), let us find a connection between (\ref{sd1-VE}) and (\ref{sd2-VE}). First, by taking a backward difference of Eq. (\ref{sd2-VE3}), we obtain
\begin{equation}\label{sd1-VE3}
  \frac{1}{\delta_l} \frac{d}{ds}(u_{l+1}-u_{l}) -\frac{1}{\delta_{l-1}} \frac{d}{ds}(u_{l}-u_{l-1})
  = \frac 32 (u_{l+1}-u_{l-1}) -\frac 14 \left(\delta^2_{l} -\delta^2_{l-1}\right)\,.
\end{equation}
On the other hand, by substituting the third equation into the first equation in (\ref{sd1-VE}) and eliminating $m_l$, one arrives at exactly the same equation (\ref{sd1-VE3}).
\begin{remark}
Although we have derived semi-discrete analogues of the reduced Ostrovsky equation (\ref{vakhnenko}) and its differentiation form (\ref{short-DP-eq}) from totally different bilinear equations, the connection between them is clear here. In other words, the semi-discrete analogue for the short wave limit of the DP equation is simply
a backward difference of semi-discrete reduced Ostrovsky equation. This finding corresponds to the fact that a differentiation of the reduced Ostrovsky equation (\ref{vakhnenko}) with respect to spatial variable $x$ gives rise to the short wave limit of the DP equation (\ref{short-DP-eq}) in the continuous case. Forward difference and differentiation are two typical operators corresponding to discrete systems and continuous systems, respectively.
In the world of integrable systems, we observe a perfect correspondence between these two operators and discrete and continuous systems.
\end{remark}
Lastly, for the sake of convenience, we list the $\tau$-functions for one- and two-soliton solutions.

\noindent {\bf One-soliton} \\
For $N=1$, we have
\begin{equation}
\fl \tau_l={\rm Pf}(1,2)=c_1+\frac{p_1-p_2}{p_1+p_2}e^{\eta_1(l)+\eta_{2}(l)}\,,
\end{equation}
where $c_1$ is a nonzero constant,
$$
e^{\eta_i(l)}=\left(\frac{1+bp_i}{1-bp_i}\right)^le^{p^{-1}_is+\xi_{i0}}\,,
$$
and $p_1$, $p_2$ are related by a constraint
\begin{equation}
\frac{1}{p_1^3}+\frac{1}{p_{2}^3}
=b^2\left(\frac{1}{p_1}+\frac{1}{p_{2}}\right)\,.
\end{equation}
\noindent {\bf Two-soliton} \\
For $N=2$, we have
\begin{eqnarray*}
\fl \tau_l&=&{\rm Pf}(1,2,3,4)
={\rm Pf}(1,2){\rm Pf}(3,4)-{\rm Pf}(1,3){\rm Pf}(2,4)
+{\rm Pf}(1,4){\rm Pf}(2,3)\\
\fl &=&\frac{p_1-p_2}{p_1+p_2}e^{\eta_1(l)+\eta_{2}(l)} \times
 \frac{p_3-p_4}{p_3+p_4}e^{\eta_3(l)+\eta_{4}(l)}
-\frac{p_1-p_3}{p_1+p_3}e^{\eta_1(l)+\eta_{3}(l)} \times
 \frac{p_2-p_4}{p_2+p_4}e^{\eta_2(l)+\eta_{4}(l)}\\
\fl &&\qquad +\left(c_1+\frac{p_1-p_4}{p_1+p_4}e^{\eta_1(l)+\eta_4(l)}\right)
\left(c_2+\frac{p_2-p_3}{p_2+p_3}e^{\eta_2(l)+\eta_3(l)}\right)\,,
\end{eqnarray*}
under the condition
\begin{equation}
\frac{1}{p_1^3}+\frac{1}{p_{4}^3}
=b^2\left(\frac{1}{p_1}+\frac{1}{p_{4}}\right)\,, \quad
\frac{1}{p_2^3}+\frac{1}{p_{3}^3}
=b^2\left(\frac{1}{p_2}+\frac{1}{p_{3}}\right)\,.
\end{equation}
Letting $c_1=c_2=1$ and $e^{\gamma_1}=\frac{p_1-p_4}{p_1+p_4}$ and
$e^{\gamma_2}=\frac{p_2-p_3}{p_2+p_3}$, the above $\tau$-function
can be rewritten as
\begin{equation}
\fl \tau_l=1+e^{\eta_1(l)+\eta_4(l)+\gamma_1}
+e^{\eta_2(l)+\eta_3(l)+\gamma_2}
+b_{12}e^{\eta_1(l)+\eta_2(l)+\eta_3(l)+\eta_4(l)+\gamma_1+\gamma_2}\,,
\end{equation}
where
\begin{equation}
\fl b_{12}=
\frac{(p_1-p_2)(p_1-p_3)(p_4-p_2)(p_4-p_3)}{(p_1+p_2)(p_1+p_3)(p_4+p_2)(p_4+p_3)}\,.
\end{equation}
In the continuous limit $b \to 0$, it is obvious that above one- and two-soliton solutions for semi-discrete reduced Ostrovsky equation converge to the one- and two-soliton solutions for the reduced Ostrovsky equation listed in \cite{FMO-VE}.
\section{Conclusion and further topics}
There are two versions of the reduced Ostrovsky equation, one is the original form (\ref{vakhnenko}), the other is its differentiation form, or is also called the short wave limit of the DP equation (\ref{short-DP-eq}). In the present paper, we have constructed their integrable semi-discretizations separately based on their different bilinear forms.
Two versions of integrable semi-discretizations of the reduced Ostrovsky equation share the same $N$-soliton solution in terms of pfaffians, which converges to the $N$-soliton solution of the continuous Ostrovsky equation (\ref{vakhnenko}), as well as its differentiation form (\ref{short-DP-eq}). The connection between two versions of integrable discretizations is made clear. In the continuous case, the short wave limit of the DP equation (\ref{short-DP-eq}) is the differentiation form of the reduced Ostrovsky equation, whereas in the discrete case, the semi-discrete short wave limit of the DP equation is the forward difference for the semi-discrete reduced Ostrovsky equation.

Similar to our previous results \cite{dCH,dCHcom,SPE_discrete1}, the semi-discrete reduced Ostrovsky equation proposed here can be served as an integrable numerical scheme, the so-called self-adaptive moving mesh method, for the numerical simulation. It seems that the semi-discrete reduced Ostrovsky equation (\ref{sd2-VE}) has more advantages than the semi-discrete analogue of the short wave limit of the DP equation in serving as a self-adaptive moving mesh method. We would like to report our results in this aspect in a forthcoming paper. Finally, we haven't succeeded in constructing an integrable fully discrete reduced Ostrovsky equation. If we could have done so, then a newly integrable discrete Tzitzeica equation might be constructed due to a direct link between these two equations. It is a further topic to be explored in the future. Another problem to be solved is the integrable discretization of the DP equation which is a more challenging problem in compared with the ones of the Camassa-Holm equation and the reduced Ostrovsky equation. We are tacking this problem based on our previous work on the DP equation \cite{FMO-DP}.
\section*{Acknowledgment}
This work of BF is partially supported by the National Natural Science Foundation of China (No. 11428102).

\end{document}